\documentclass[11pt]{article}

\usepackage{algos}

\begin{document}

\title{Algorithms for Noisy Broadcast under Erasures}
\author{Ofer Grossman \footnotemark[1]
\and Bernhard Haeupler \footnotemark[2]
 \and Sidhanth Mohanty \footnotemark[3]}
\footnotetext[1]{EECS Department, MIT, Cambridge, MA, USA. \texttt{ofer.grossman@gmail.com}.
Supported by the Simons Grant -- Simons Foundation, ``Investigation Award -- Goldwasser'' 08/01/2012 -- 12/31/2017, NSF Grant -- CNS-1413920, and a Hertz Foundation Fellowship.}
\footnotetext[2]{Computer Science Department, Carnegie Mellon University, Pittsburgh, PA, USA. \texttt{haeupler@cs.cmu.edu}. 
Supported in part by NSF grants CCF-1527110, CCF-1618280 and NSF CAREER award CCF-1750808.}
\footnotetext[3]{Computer Science Department, Carnegie Mellon University, Pittsburgh, PA, USA. \texttt{sidhanthm96@gmail.com}}
\maketitle

\pagenumbering{gobble}

\begin{abstract}
The noisy broadcast model was first studied in [Gallager, TranInf'88]
where an $n$-character input is distributed among $n$ processors, so that each processor receives one input bit. Computation proceeds in rounds, where in each round each processor broadcasts a single character, and each reception is corrupted
independently at random with some probability $p$. [Gallager, TranInf'88]
gave an algorithm for all processors to learn the input in $O(\log\log n)$ rounds with high probability.
Later, a matching lower bound of $\Omega(\log\log n)$ was given in [Goyal, Kindler, Saks; 
SICOMP'08].

We study a relaxed version of this model where each reception is erased and replaced
with a `?' independently with probability $p$. In this relaxed model, we break past the lower bound of [Goyal, Kindler, Saks; 
SICOMP'08] 
and obtain an $O(\log^* n)$-round algorithm for all processors to learn
the input with high probability. We also show an $O(1)$-round algorithm for the same problem when the
alphabet size is $\Omega(\mathrm{poly}(n))$.
\end{abstract}
\newpage
\pagenumbering{arabic}

\section{Introduction}
In recent years, it is becoming increasingly common for computational tasks to be
performed by multiple processors in a distributed fashion. The communication channels of these networks may have imperfections,
which introduces noise to the system.

A formal version of a noise model was proposed by \cite{el1987open}: There are $n$ processors: $1,2,\ldots, n$ and each processor is given
a bit. In each round, every processor broadcasts a bit to all other processors. Every processor
will receive the correct message with some probability, and may receive a different (corrupted) message independently with probability
$p < 1/2$ (i.e., each reception gets corrupted with probability $p$). The goal is for the processors to collectively compute the {XOR} of all
their inputs. An algorithm that takes $O(\log\log n)$ rounds for all processors to learn
the full input (and hence the XOR as well) was found by \cite{gallager1988finding}. A matching lower bound of
$\Omega(\log\log n)$ rounds was proven by \cite{goyal2008lower}.

All of the prior works were concerned with \textit{substitution errors}. In this paper, we study such networks in the presence of \textit{erasure errors}, where instead of messages getting corrupted into other messages, instead messages may get dropped. Specifically, we study the following model: in a single round each processor
can broadcast a single bit $b$ to all other
processors. For each ordered pair $(i,j)$, independently with some probability $p$, the character
that $i$ transmitted is not received by $j$ and a `?' is received instead. In other words,
there is a string $X\in\{0,1\}^n$ and processor $i$ is given the $i$th bit of $X$, called
$x_i$, and the goal is for each processor to learn $X$ using as few rounds of communication as possible. We call our noise model the \textit{erasure model}.

\subsection{Our results}
We show that for any alphabet, each processor can learn the inputs of all other processors with
high probability within $O(\log^* n)$ rounds. At the high level, the algorithm involves recursively
running the protocol on groups of size $\log n$, and having each group encode
its input using a constant rate and constant relative distance
error correcting code. Then, the group collectively transmits this encoded string within a constant number of rounds. It can be shown
that with high probability every processor receives enough bits to decode the group's input. There are groups
for which not enough processors learn the full string (i.e., the recursive call fails), and some technical steps are
needed to handle these `failed groups'. The protocol is described in full
detail in section \ref{sec:mainres}.

We note that in the presence of \textit{substitution errors}, it was proven in \cite{goyal2008lower} that $\Omega(\log \log n)$ rounds are required for all processors to learn the whole input. Since we show a $O(\log^*(n))$ algorithm for the problem in the presence of \textit{erasure errors}, this shows a fundamental difference between substitution errors and erasure errors in the broadcast model.

We then show that when the alphabet is of polynomial size, there is an $O(1)$ round
algorithm for every processor to learn the full input. The algorithm involves
treating the alphabets as elements of a finite field $\mathbb{F}_q$, and simulating
multiplying the input vector with an appropriate random matrix. Then , the processors receive a random system of linear equations which one can show has a
unique solution with high probability.

We then show that any symmetric function of the input
can be computed within a constant number of rounds via
computing the Hamming weight.

\subsection{Related Work}
A related problem was studied in \cite{gallager1988finding} where the broadcast
model assumed was sequential, where in one round only one processor can broadcast a bit. Additionally, the noise model assumed was that of bit flips instead of
erasures. That is, each transmitted bit is independently flipped with
probability $p$ on the receiving end. In their model, \cite{gallager1988finding} shows that all
the processors can learn the entire input within $O(\log\log n)$ rounds. However,
it left open the question of whether a faster protocol was possible.

The model of \cite{gallager1988finding} was studied further in \cite{goyal2008lower} where a lower bound of
$\Omega(n\log\log n)$ was proven for the total number of broadcasts, thereby establishing
that Gallager's protocol is optimal up to constant factors. The lower bound is proved via a reduction to
another model called the generalized noisy decision tree, which is a variant of the
noisy decision tree model introduced in \cite{feige1994computing}.
\cite{goyal2008lower} also studies whether more efficient protocols exist when
the processors only want to compute some specific function on the entire input and
shows that the Hamming weight can be computed with constant probability within $O(n)$
broadcasts.

We note that it follows from the lower bound in \cite{goyal2008lower}
that in a variant of our model where one considers substitution errors instead of erasure errors, any
protocol from which all the processors learn the entire input must take
$\Omega(\log\log n)$ rounds. In light of this lower bound, our result of an
$O(\log^* n)$ protocol is interesting, as it shows a fundamental difference between substitution and erasure errors in this broadcast model.

Recently, a work by Efremenko, Kol, and Saxena \cite{efrem} showed that under a model where the processors can adaptively choose which processor will speak in each round, the lower bound of \cite{goyal2008lower} breaks down.

Note that the work of Gallager \cite{gallager1988finding} shows that in the
substitution model where a single processor broadcasts to the rest in a round, any
function can be computed within $O(n\log\log n)$ rounds. A work by Kushilevitz and
Mansour \cite{kushilevitz1998computation} studies the question of which Boolean functions can be computed within $O(n\log\log n)$
broadcasts. They determine that threshold functions can be computed with constant probability within $O(n)$ broadcasts.

A paper by Feige and Killian \cite{feige2000finding} studied a harsher noise
model than \cite{gallager1988finding}, where an adversary can arbitrarily
`uncorrupt' arbitrary corrupted bits, causing the noise to lose structure. In this harsher model, they show an $O(\log^*n)$ round protocol to compute
the \textsf{OR} of all input bits. Newman \cite{newman2004computing} studies
another noise model where each bit transmitted is independently flipped with
an unknown probability that is at most $p$ and gives algorithms that use $O(n)$
broadcasts and $O(\log^* n)$ rounds for certain classes of Boolean functions,
including \textsf{OR}, \textsf{AND}, and functions with linear size $\mathsf{AC_0}$
formulas.

In \cite{alon}, the authors show efficient protocols to handle errors in the UCAST model, in which instead of broadcasting bits, a processor can send a different message to each other processor. They also show efficient protocols to handle errors when the communication network has certain expansion properties. For general graphs of low degree, a protocol for handling errors was found in \cite{raj}, which was later shown to be optimal in \cite{braver}.

Our model in the absence of errors is known as the Broadcast Congested Clique, which is a computational model often studied in distributed computing (see for example, \cite{bcast1, bcast2, bcast3,bcast4,bcast5}. In this model, $n$ processors each get a piece of the input, and they work together to compute some function of this shared input. Computation proceeds in rounds, where in each round each processor can broadcast a short message to all other processors. Our work can be interpreted as showing that when using messages of constant size, every protocol in the Broadcast Congested Clique can be made resilient to erasure errors with a blowup of only $O(\log^*(n))$. In the case where messages are of logarithmic size, we show the Broadcast Congested Clique can be made resilient to erasure errors with only a constant blowup.


\subsection{Notation and conventions}\label{sec:prelim}

In this section, we state some notational conventions we use. First, we describe the computational model (without erasures), and then we formally define the model we consider with erasures.

\textbf{The Computational Model:} In a setting with $n$ processors, each processor is identified with a distinct
number in $[n]$. Given a string $X$, which we denote using an upper case
character, we write the $i$th bit as $x_i$, using the corresponding lower case character. To denote 
the substring of $X$ starting at position $i$ and ending in position $j$ we write $X_{[i,j]}$.
When we wish to compute some function of a $n$-bit string $X$ using $n$ processors, assume
$x_i$ is provided as input to processor $i$.
In the description of algorithms, $\textsc{Algo}(x_1,\ldots,x_n)$ refers to an algorithm
that runs on $n$ processors where the $i$th processor is given $x_i$ as input.

In all our algorithms, we assume that each broadcast is repeated $\gamma$ times where $\gamma$
is some appropriately chosen constant.

Formally, we have:
\begin{definition}
We let the \textit{noisy parallel broadcast model} be a model of computation where there are $n$ processors $P_1, \ldots, P_n$, and $P_i$ receives input bit $x_i$. In each round of computation, each processor can broadcast one bit to all other processors. Each reception is corrupted with some constant probability $0\le p < 1$, in which case the character `?' is received instead of the bit which was sent.
\end{definition}

In this paper, we study the complexity of computing certain functions in the above model. Specifically, for constant erasure probability $p$ we show a bound of $O(\log^*(n))$ for computing any function, and a bound of $O(1)$ for symmetric functions.

As part of our algorithm we use error correcting codes, so we include standard results and notations for codes below:
\textbf{Error Correcting Codes:} An error correcting code is described by functions $\mathsf{Enc}:\{0,1\}^k
\rightarrow\{0,1\}^{n}$ and $\mathsf{Dec}:\{0,1\}^{n}\rightarrow
\{0,1\}^k$.

The \textit{rate} of an error correcting code is defined as $\frac{n}{k}$
and the \textit{relative distance} is defined as $\frac{\min_{x,y\in\mathcal{C}}
d(x,y)}{n}$. The quantity $\frac{d(x,y)}{2}$ is referred to as the \textit{decoding
radius}. The decoding function $\mathsf{Dec}:\{0,1\}^{n}\rightarrow
\{0,1\}^k$ satisfies the property that $\mathsf{Dec}(c') = y$ for any $c'$ within hamming distance $\frac{d(x,y)}{2}$ (i.e., the decoding radius) from $\mathsf{Enc}(x)$.

We use the result of \cite{justesen1972class} that error correcting code
families of constant rate and constant relative distance exist.
In particular, for the sake of this paper, we assume the existence of an error
correcting code family $E$ with relative distance $0.25$ and rate some absolute
constant $K$.

\section{An $O(\log^* n)$ algorithm for computing any function}\label{sec:mainres}

We consider the following message-passing model. There are $n$ processors, and in each round, every processor transmits a single bit $b$
to all other processors. Each processor receives each bit independently and at random with probability $1-p$. With probability $p$, the 
character `?' is received instead. If each processor starts with a single input bit, we ask how many rounds are required so that every 
processor knows all input bits with high probability. We show a bound of $O(\log^*(n))$ for this problem. Specifically, we will show:

\begin{theorem}\label{main}
For every $0\le p<1$, there is an algorithm in the noisy broadcast parallel erasure model that computes $\text{ID}_n$ with high probability within $O(\log^*(n) \log \frac{1}{1-p})$ rounds.
\end{theorem}

Without loss of generality, we assume that $p\leq 0.01$, since for any erasure probability  $p < 1$, repeating each message $O(\log\frac{1}{1-p})$ times can be used to
effectively lower the probability of receiving `?'. We describe our algorithm for the case where the alphabet $\Sigma=\{0, 1\}$.
The protocol generalizes to larger alphabets in a straightforward manner.

We describe a protocol for $n$ processors with the guarantees: at the end of
the protocol, all $n$ processors can output the full string $C$ with probability at least
$1-\frac{1}{n^5}$, and if the protocol fails (that is, there is
some processor who cannot output the full string $C$), then all $n$ processors
can output `$\bot$' with probability at least $1-\frac{1}{2^{7n}}$. For the rest of this section,
we assume $n\geq n_0$ for a sufficiently large $n_0$.

We begin by describing algorithms for simpler subproblems.

\begin{lemma}
\label{lem:andconst}
Let $b_i$ be the input to processor $i$, and let the erasure probability $p$ be $.01$. Then there is an $O(1)$-round algorithm
and an absolute constant $\alpha$ such that all processors output the \textsf{AND}
of all $b_i$ with probability at least $1-2^{-\alpha n}$.
\end{lemma}

\begin{proof}
\textbf{Algorithm:} The algorithm is as follows: in each round, a processor $i$ broadcasts `0' either if $x_i = 0$ or if processor $i$ has received at least one `0' in at least one of the previous rounds. Otherwise, processor $i$ broadcasts $1$. This is repeated for $100$ rounds.

Processor $j$'s output is the AND of all bits it received.

\textbf{Analysis:}
First, note that if all of the $b_i = 1$, then all processors must output $1$, no matter what messages were corrupted, since all received bits of all processors must be 1.

Now, suppose there is an $i$ for which $b_i=0$. Let $t$ be
the number of processors that received the transmission of $i$ in the
first round. The probability that processor $j$ receives only 1s in the
second round is at most $p^{t}$.

We can use Hoeffding's inequality to obtain
\[
\Pr\left[t < \frac{n}{2}\right]\leq e^{-\alpha'n}
\]
for some constant $\alpha'$.
Thus, the probability that there is some $j$ that received only 1's even
if there is a processor with a 0 is at most $n(e^{-\alpha'n}+p^{n/2})$,
bounded above by $e^{-\alpha n}$ for a constant $\alpha$.
\end{proof}

We note that the above protocol does not work in the substitution model (the model where a message may be flipped with small
probability, as opposed to being corrupted to a `?'). In fact, in \cite{goyal2008lower} it was proven that computing the AND function with high probability in the substitution model requires
$\Omega(\log \log n)$ rounds.

We next show an $O(1)$ round algorithm for \textsc{Equality Testing}. Each processor is given an $n$-bit string $S_i$
as input, and the goal is for all processors to output 1 if all their inputs are equal and 0 otherwise with probability at least
$1-2^{-\Omega(n)}$. Unless otherwise specified, each step of the algorithm is from the view of processor $i$. Roughly speaking, this step will be used in the main algorithm to verify that all processors end up with the same output string $S$.

\begin{algorithm}[H]
\caption{\textsc{EqualityTest}$(S_1,\ldots,S_n)$}\label{algo:equalitytest}
\normalsize{
$\mathsf{Enc}$ is the encoding function of a code $\mathcal{C}$ with relative distance $0.25$ and constant
rate $K$.

\hrulefill

\begin{enumerate}
\item \label{step:transmissioncode}
Transmit $(\mathsf{Enc}(S_i))_{[(i-1)K+1,iK]}$ over $K$ rounds

\item
Let $A_{t,i}$ be the $K$-bit string received from processor $t$ and $A_i=A_{1,i}A_{2,i}\ldots A_{n,i}$. Set $c_i$ to
1 if Hamming distance between $A_i$ and $\mathsf{Enc}(S_i)$ is at most $0.06Kn$ and 0 otherwise

\item The processors run the AND protocol from Lemma \ref{lem:andconst} and output the \textsf{AND}
of all $c_i$
\end{enumerate}
}
\end{algorithm}

\begin{lemma}
\label{lem:equalconst}
When the erasure probability $p \le .01$, Algorithm \ref{algo:equalitytest} correctly solves \textsc{Equality Testing} with probability
at least $1-2^{-\beta n}$ for some absolute constant $\beta$.
\end{lemma}

\begin{proof} Let $A$ be the string collectively transmitted by all processors in Step \ref{step:transmissioncode}. 
We know
\begin{align*}
d(\mathsf{Enc}(S_i),\mathsf{Enc}(S_j)) &\leq d(\mathsf{Enc}(S_i), A_i) + d(A_i, A) + d(A, A_j) + d(A_j,\mathsf{Enc}(S_j))\\
&\leq 2(d(\mathsf{Enc}(S_i), A_i) + d(\mathsf{Enc}(S_j), A_j))
\end{align*}
where the second inequality is because $d(A,A_i)$ is the number of `?'s received, and lower bounds $d(\mathsf{Enc}(S_i),A_i)$.

If both $d(\mathsf{Enc}(S_i), A_i)$ and $d(\mathsf{Enc}(S_j),A_j)$ are at most $0.06Kn$, then $d(\mathsf{Enc}(S_i),
\mathsf{Enc}(S_j))$ is at most $0.24Kn$, but since they are codewords of a code with relative distance $0.25$,
$\mathsf{Enc}(S_i)=\mathsf{Enc}(S_j)$, implying $S_i=S_j$. So if there is a pair $i,j$ with $S_i\neq S_j$, then either $c_i$ or $c_j$
must be 0. And then from Lemma \ref{lem:andconst}, with probability at least $1-e^{-\alpha n}$, the processors correctly
detect that there is a $c_i$ equal to 0.

On the other hand, if all the strings are indeed equal, then $c_j$ is 0 only if processor $j$ receives
fewer than $0.94Kn$ bits. We upper bound the probability that this happens 
by using Chernoff bound along with a union bound over all processors.
\begin{align*}
n\Pr[\text{processor }i\text{ receives fewer than }0.88Kn\text{ bits}]\leq ne^{-\alpha''n} \le e^{-\alpha' n}
\end{align*}
where $\alpha'$ is some constant. We let $\beta=\min\{\alpha,\alpha'\}$.
\end{proof}

Let $X=x_1x_2\ldots x_n$ be the input string and processor $i$ is given $x_i$ and is
required to output a tuple $(X_i,s_i)$, where $X_i$ an $n$-bit string and $s$ either
1, indicating success or 0, indicating failure, with the goal of having all $X_i=X$ and all $s_i =1$.
We say that an algorithm on a group of processors \textbf{succeeded} if
$X_i=X$ and $s_i=1$ for all $i$, \textbf{failed with knowledge} if $r_i=0$ for all $i$,
and \textbf{failed without knowledge} otherwise. We describe an algorithm for this
problem where each step is from the view of processor $i$ unless otherwise specified.
Recall that each broadcast is repeated $\gamma$ times to effectively reduce the erasure probability $p$ to be at most $.01$. For simplicity, we assume that $n$ is a power of 2, and so $\log n$ is an integer. It is easy to generalize the algorithm to all values of $n$.

At the high level, the algorithm proceeds as follows. We partition the processors into $n/\log n$ sets of size $\log n$ each (Step \ref{step:recurse}). Then, we recursively compute the input on each of these subsets. Now, some of these subsets will have succeeded, and some will have failed. For the ones that failed, we now recompute the input, but this time we add more processors to be ``helper processors". That is, the processors which succeeded in the recursive calls will now be used to aid the processors who failed in the recursive call by sending messages on their behalf. This can be seen in Step \ref{step:amplification}, where the processor sends $x_{\ell_i}$, which is the input to a processor which failed on the recursive call. This idea of using successful processors to help others who failed helps ensure that within a constant number of tries, with high probability all input bits will be known.

\begin{algorithm}[H]
\caption{\textsc{LearnInput}$(x_1,\ldots,x_n)$}\label{algo:learninput}
\normalsize{
$\mathsf{Enc}$ is the encoding function of a code $\mathcal{C}$ with relative distance $0.25$ and constant
rate $K$

\hrulefill

\begin{enumerate}

\item \textbf{Base Case:} If $n < 100$
\begin{enumerate}
\item Transmit $x_i$ repeatedly $100$ times, and set string $S_i$
as per
\[
(S_i)_j =
\begin{cases}
b & \text{if $b$ was received in any transmission from $j$}\\
\text{random bit} & \text{if all transmissions from $j$ are `?'}
\end{cases}
\]
and go to Step \ref{step:verification}.
\end{enumerate}

\item \textbf{Recursive Step:} 
\begin{enumerate}
\item \label{step:recurse} Recursively obtain $(X_i',r_i') = \textsc{LearnInput}\left(x_{\left\lfloor\frac{i}{\log n}\right\rfloor\log n+1},\ldots,
x_{\left\lfloor\frac{i}{\log n}\right\rfloor\log n+\log n}\right)$. We call this set of processors the \textit{group} of $i$. 

\item Broadcast $r_i'$

\item Set
$R_i$ by setting $(R_i)_j$ to 1 if only 1's were received from $j$'s group (i.e., from the processors which
$j$ computed the recursive call with) and 0 otherwise, for each $j\in[n]$.

\item \label{step:encoding} Let $i'=i~\mathrm{mod}\log n$ and transmit
$\mathsf{Enc}(X_i')_{[(i'-1)K+1,i'K]}$ over the next $K$ rounds.

\item Let $z_i$ be the number of zeros in $R_i$ and let $j =
\left\lceil\frac{iz_i}{n}\right\rceil$ and let $\ell_i$ be the index of
the $j$th zero in $R_i$.
Create set $M_{s,i}$ to be all $t$ such that
\[\frac{n(s-1)}{z_i}<t\leq\frac{ns}{z_i}\]

\item \label{step:announcement} Transmit $x_i$.

\item \label{step:amplification} Broadcast what was received
from $\ell_i$, which is either `?' or $x_{\ell_i}$. Let $M'_{j,i}$
be the set of characters received from $M_{j,i}$.

\item Set $X_i$ by setting $(X_i)_j$ to $x_i$ if $j=i$, by decoding the bits received in Step \ref{step:encoding} if $(R_i)_j=1$
and at least $0.88K\log n$ bits were received from the group of $j$, to a random bit if $(R_i)_j=1$ and fewer than $0.88K\log n$
bits were received from group $j$ in Step \ref{step:encoding}, and to $\Ind_{1\in M'_{j,i}}$ if $(R_i)_j=0$.
Proceed to Step \ref{step:verification}.

\end{enumerate}

\item \textbf{Verification of output}
\begin{enumerate}
\item \label{step:verification} Obtain $v_i=\textsc{EqualityTest}(X_1,\ldots,X_n)$ and output $(X_i,v_i)$.
\end{enumerate}

\end{enumerate}
}
\end{algorithm}

We now prove the following proposition, from which Theorem \ref{main} immediately follows.
\begin{proposition}
\label{thm:main}
Algorithm \ref{algo:learninput} runs in $O(\log^* n)$ rounds, succeeds (i.e., each processor outputs $(X, 1)$, where $X$ is the input to all processors) with probability at least $1-\frac{1}{n^5}$
and fails without knowledge with probability at most $\frac{1}{2^{7n}}$.
\end{proposition}
\begin{proof}
We list conditions under which the protocol definitely succeeds, and show all these conditions
hold with probability at least $1-\frac{1}{n^5}$. Define $R$ as $r_1'r_2'\ldots r_n'$ from the
output of Step \ref{step:recurse}. Define $M_s$ as all $j$ such that $\frac{n(s-1)}{z}<j\leq\frac{ns}{z}$
where $z$ is the number of 0's in $R$.

The protocol definitely succeeds if the following conditions hold:
\begin{enumerate}

\item \label{condition:boundedfailure} All but at most $\frac{n}{\log^3 n}$ groups succeed in the recursive call of Step \ref{step:recurse}.

\item \label{condition:nocatastrophe} No group fails without knowledge in the recursive call of Step \ref{step:recurse}, and $R_i=R$ for all $i$.

\item \label{condition:amplification} For all $j$ such that $(R)_j=0$, for all $i$, processor $i$ receives at least one transmission
from a processor in $M_{\ell_j,i}$ in Step \ref{step:amplification} where the $\ell_j$th 0 in $R$ occurs at $(R)_j$.

\item \label{condition:decodingradius} Each processor receives at least $0.88K\log n$ bits from
each successful group in at least one transmission in Step \ref{step:encoding} of the algorithm.

\end{enumerate}

Indeed, for any $j$ in a successful group, all processors correctly learn the input to processor $j$ because Condition
\ref{condition:decodingradius} is met. By Condition \ref{condition:nocatastrophe}, for fixed $s$, $M_{s,i}$ is the same for all $i$ since $M_{s,i}$ depends on $R_i$. For any $j$ in a failed group, by Condition \ref{condition:nocatastrophe}, $(R)_j=(R_i)_j=0$, and by Condition \ref{condition:amplification}, each processor receives at least one transmission of processor $j$'s input in Step \ref{step:amplification} and so all processors correctly learn the input to processor $j$.

We now proceed with showing a lower bound on the probability that all of these conditions hold.

We can see that Condition \ref{condition:boundedfailure} holds with probability at least $1-\frac{1}{n^6}$ since by Chernoff bounds, the number of failed groups exceeds $\frac{n}{\log^3 n}$ with probability at most $\frac{1}{n^6}$.

Now suppose Condition \ref{condition:boundedfailure} holds.
A group fails without knowledge with probability at most $\frac{1}{n^7}$ by the guarantees of the
protocol. The probability that there exists a group that failed without knowledge, by the union bound,
is therefore at most $\frac{1}{n^6}$. If no group failed without knowledge,
the only way $R_i$ cannot equal $R$ is if there is a group $M_{j, i}$ that processor $i$ did not receive a single bit from.
The probability that processor $i$ does not receive a single bit from this group is
$p^{\gamma\log n}$, which for appropriate $\gamma$ is at most $\frac{1}{n^8}$. Thus, the probability that there
is some $i,j$ pair such that processor $i$ does not receive a single bit from group $j$ is at most $\frac{1}{n^6}$
by a union bound.
So the probability that Condition \ref{condition:nocatastrophe} is not met (given that Condition \ref{condition:boundedfailure} is met) is at most $\frac{2}{n^6}$.

Note that $R_i=R$ means $M_s=M_{s,i}$ for all $i$. It follows from Chernoff bounds that the number of processors
in $M_s$ that receive the bit transmitted by processor $s$ is at least $\log n$ with probability at least
$1-\frac{1}{n^6}$. The probability that processor $i$ does not receive any bits from processors in $M_s$ in
any of the repetitions of Step \ref{step:amplification} is at most $p^{\gamma\log n}$, which can be
made smaller than $\frac{1}{n^8}$ by setting $\gamma$ to be large enough. Now by taking a union bound over all pairs $(i,s)$
we can conclude that Condition \ref{condition:amplification} does not hold with probability at most
$\frac{1}{n^6}$.

The probability that processor $i$ receives fewer than $0.88K\log n$ bits from group $j$ in all repetitions of
Step \ref{step:encoding} is at most $\frac{1}{n^{c\gamma}}$ for some constant $c$ by Chernoff bounds.
A union bound across all processor-group pairs tells us that Condition \ref{condition:decodingradius} does not hold with probability at
most $\frac{1}{n^{c\gamma-2}}$ which can be made smaller than $\frac{1}{n^6}$ with large enough $\gamma$.

Based on the bounds we obtained on the probability that each of Conditions \ref{condition:boundedfailure}, \ref{condition:nocatastrophe},
\ref{condition:amplification}, \ref{condition:decodingradius} don't hold, we can conclude that the probability that all the
conditions hold is at least $1-\frac{1}{n^5}$.

It remains to show that the probability that the processors failed without knowledge is at most $2^{-7n}$. If
there is $X_i$ such that $X_i\neq X$, then it differs from $X$ in some index $j$, which means $(X_i)_j
\neq (X_j)_j$ by construction of $X_j$ implying $X_i\neq X_j$. Thus, a failure without knowledge happens
only if Step \ref{step:verification} fails, which happens with probability at most $e^{-\beta n}$, which can
be made smaller than $2^{-7n}$ by choosing the number of repetitions $\gamma$ to be a large enough constant.

The number of rounds this algorithm takes is given by $T(n)$, which satisfies the recurrence relation $T(n)=T(\log n)+L$ where $L$ is a constant and
with base case $T(100) = O(1)$, which solves to $T(n)=O(\log^* n)$.
\end{proof}

\section{An $O(1)$ algorithm for large alphabets}\label{sec:largealph}
For large alphabets, in the regime where the alphabet $\Sigma$ is $\mathbb{F}_q$ and
$q = \mathrm{poly}(n)$, we give a constant round algorithm to have all processors learn the input
$X$ with probability at least $1-\frac{1}{\mathrm{poly}(n)}$. Unless otherwise specified,
the algorithm is from the view of processor $i$. While our algorithm works for any $q$ that
is polynomial in $n$, for simplicity of exposition we assume $q\geq n^6$ and that $q$ is a prime.

\begin{algorithm}
\caption{\textsc{LearnInputLargeAlphabet}$(x_1,\ldots,x_n)$}\label{algo:learnlargealph}
\normalsize{
Let $F$ be a function that encodes subsets of $[6\log n]$ as elements of $\mathbb{F}_q$

\hrulefill

\begin{enumerate} 

\item Let $k = \lfloor 6\log n \rfloor $ and determine $B_i = \{\frac{nj}{k}+1,\ldots,
\frac{n(j+1)}{k}\}$, where $j$ is chosen such that $i\in B_i$

\item \label{step:firsttransmit} Broadcast $x_i$ for $10$ rounds

\item For each $t$ from 1 to $10$ and for each processor in $B_i$ from
which an entry was received in round $t$ of Step \ref{step:firsttransmit},
choose the processor with probability $\frac{1}{2(1-p)}$ and choose $i$
with probability $\frac{1}{2}$. Let $T_{t,i}$ be the set of
chosen elements.

\item For the next $20$ rounds, processor $i$ transmits all the
$\sum_{b\in T_{t,i}}x_b$ (where the $x_b$ are added as elements of $\mathbb{F}_q$) and $F(T_{t,i})$

\item Output $X_i$ consistent with all received pairs
$\left(\sum_{b\in T_{t,i}}x_b,F(T_{t,i})\right)$. If there is more than one possibility for such an $X_i$, pick one at random.

\end{enumerate}
}
\end{algorithm}

\begin{theorem}\label{thm:largealph}
With probability at least $1-\frac{1}{\mathrm{poly}(n)}$, after running Algorithm \ref{algo:learnlargealph}, all processors will know all other processors' inputs. Furthermore, the algorithm terminates within $O(1)$ rounds.
\end{theorem}

As a first ingredient towards proving Theorem \ref{thm:largealph}, we prove the following lemma.
\begin{lemma}
\label{lem:randomrank}
If $A$ is a $5k\times k$ random binary matrix where each entry is i.i.d.
generated by flipping a fair coin, then with probability at least $1-e^{-0.4k}$,
$A$ is full rank.
\end{lemma}
\begin{proof}
Suppose $V$ is a subspace of $\mathbb{F}_q^{k}$ that is not equal to all of
$\mathbb{F}_q^{k}$, then we can find standard basis vector $e_i$
that is not in $V$. Then for any binary vector $v$, consider
$v'$ with the bit at the $i$-th
coordinate flipped. Either $v$ or $v'$ is not in $V$, which
means at least half of the binary vectors are not in
$V$, which means each new vector has probability at least
$\frac{1}{2}$ of not being in $V$. If we let
$V=\mathsf{span}\{\text{vectors drawn so far}\}$,
then each draw has a probability at least $\frac{1}{2}$ of
increasing the dimension.
Suppose we flip $5k$ coins, the probability that the number of heads is at most
$k$ is an upper bound on the probability of the span of $5k$ randomly drawn
vectors not being the whole space.

By Chernoff bounds, this probability is at most $e^{-0.4k}$.
\end{proof}

\begin{proof}[Proof of Theorem \ref{thm:largealph}]

Each $T_{t,i}$ is a uniformly random subset of input bits
of set $B_i$. Let $x_{B_i}$ be a $k$-dimensional vector of the
inputs to processors in $B_i$, then the transmitted characters in
Round 5 are of the form $(\langle a_{B_i}, x_{B_i}\rangle, F(T_{t, i}))$ where $a_{B_i}$
is a random binary vector, and $F(T_{t, i})$ is an encoding of $a_{B_i}$.
The transmitted characters can be viewed as elements in the
vector $Ax_{B_i}$, where $A$ is a matrix whose rows are the $a_{B_i}$. A single processor's output of $x_{B_i}$ is
given by sampling rows of the equation $Ay_{B_i}=Ax_{B_i}$ where
$y_{B_i}$ is indeterminate and solving for $y_{B_i}$. If the
number of sampled rows is at least $5k$, then from
Lemma \ref{lem:randomrank} the probability that the sampled
rows span $\mathbb{F}_q^k$ and hence give a unique solution to
$y_{B_i}$ is at least $1-\frac{1}{n^{2.4}}$.

The probability that the number of sampled rows for a group is
less than $5k$ can be upper bounded by $\frac{1}{n^5}$
using Chernoff bounds.

So by union bound over all group-processor pairs (i.e., all pairs $(i, B_j)$), we get
a $\frac{1}{\mathrm{poly}(n)}$ upper bound on the failure probability.
\end{proof}

\bibliographystyle{alpha}
\bibliography{report}

\appendix
\section{An $O(1)$ protocol for computing any symmetric function}\label{sec:symfunc}

We show that any symmetric function can be computed within $O(1)$ rounds in the model.
Symmetric functions are functions whose value doesn't change under permutation of the
input bits. In other words, these functions only depend on the Hamming weight of
the input string. Hence, an algorithm for every processor to learn the Hamming
weight of the string leads to an algorithm to compute any symmetric function. Our algorithm is inspired by a similar algorithm (for a different model) of \cite{goyal2008lower}.

\begin{theorem}\label{thm:symmetricprotocol}
There is an $O(1)$ round algorithm in the noisy broadcast parallel erasure model that computes
\textsf{Hamming Weight}$(X)$ with probability at least $0.75$.
\end{theorem}

Our algorithm proceeds in two phases:
\begin{enumerate}
\item Divide the interval $[0,n]$ into subintervals of length $c\sqrt{n}$ and find
which interval the Hamming weight belongs to.
\item Figure out exactly which integer in the interval is the Hamming weight.
\end{enumerate}

More precisely, the first step will give us three intervals, and we will show for at least two of these intervals, with high probability all processors will end up with the same interval. Then, we will run the second step (where we pinpoint the exact hamming weight) on each of the three intervals, and take a majority vote to compute the final output.

We describe the first step below:
\begin{algorithm}
\caption{\textsc{DetermineInterval}$(x_1,\ldots,x_n)$}\label{algo:determineinterval}
\normalsize{
$A_1,A_2,\ldots,A_k$ are disjoint intervals of size $\approx 2t\sqrt{n}$
covering $[0,n]$, with $t$ chosen later.

Let $A_i$ be $\varnothing$ if $i<1$ or $i>k$.

$B_i:=A_i\cup A_{i+1}\cup A_{i+2}$.

$\mathcal{B}_s:=\{B_i:i\equiv s~\mathrm{mod}~3\}$.

$\mathsf{Enc}$ is the encoding function of a code with relative distance $0.25$ and constant
rate $K$.

\hrulefill

\begin{enumerate} 
\item Transmit $x_i$

\item Compute $h_i:=\frac{\text{number of 1's received}}{1-p}$

\item For $s=0,1,2$:
\begin{enumerate}

\item Find interval in $\mathcal{B}_s$ containing $h_i$, called $I$. $I$ is encoded as a string $s_I$ (of size $O(\log n)$).

\item \label{step:encodeinterval} Let $i'=i~\mathrm{mod}\log n$ and transmit $\mathsf{Enc}(s_I)_{[K(i'-1) + 1, Ki']}$
over $K$ rounds

\item $C_{i,s}:=
\begin{cases}
s_I &\text{if at least $.88K\log n$ bits were  received in Step \ref{step:encodeinterval}}\\
\text{decoded string} &\text{if fewer than $.88K\log n$ bits were  received in Step \ref{step:encodeinterval}}
\end{cases}$

\end{enumerate}
\item \label{step:returninterval} Return $C_{i,0},C_{i,1}$ and $C_{i,2}$.

\end{enumerate}
}
\end{algorithm}

\begin{lemma}
With probability at least $1-\exp(-\Omega(n))$, for at least two $t$ in $\{1,2,3\}$, all $C_{i,t}$
outputted in Step \ref{step:returninterval} of Algorithm \ref{algo:determineinterval} are
equal and correspond to an interval containing $\textsf{Hamming Weight}(X)$.
\end{lemma}
\begin{proof}
By Chernoff bounds, the probability that $h_i$ deviates from the
truth by $t\sqrt{n}$ is at most $e^{-Ct^2}$ for an absolute constant $C$.
This can be made smaller than $0.01$ with appropriate choice of a constant $t$. Then for at
least two values of $s$, $h_i$ lies in the correct interval in $\mathcal{B}_s$ with
probability at least $0.99$. Without loss of generality, say this happens for
$s=0$ and $s=1$. Using Chernoff bounds, we can show that for some constant $c$,
with probability at least $1-\exp(-\Omega(n))$, at least $0.95$
fraction of the processors decode the correct interval in $\mathcal{B}_0$ and
$\mathcal{B}_1$.

And assuming at least $0.95$ fraction of the processors decode the correct intervals
in $\mathcal{B}_0$ and $\mathcal{B}_1$, we can show once again using Chernoff
bounds and union bound, that the number of bits from the encoded string of the
correct interval received by each processor is more than $0.9Kn$ with probability
at least $1-\exp(-\Omega(n))$, which means with exponentially high probability,
every processor decodes the correct interval in $\mathcal{B}_0$ and $\mathcal{B}_1$.
\end{proof}

For the second step, our goal is the
following: given that every processor knows an interval $[a,b]$ in which the
Hamming weight of the input string lies, it can recover the value of the Hamming
weight in $O(1)$ rounds.

\begin{algorithm}
\caption{\textsc{PinpointWeight}$(x_1,x_2,\ldots,x_n; [a,b])$}\label{algo:pinpointweight}
\normalsize{
$[a,b]$ is the interval of length up to $3\sqrt{n}$ where the Hamming weight is promised to lie

$\mathsf{Enc}$ is the encoding function of a code $\mathcal{C}$ with relative distance $0.25$ that maps
$\log n$ bit strings to $K\log n$ bit strings

Let $\theta_s$ be defined as the probability that when flipping $s$ coins, each coming up heads with probability $1-p$, at least $(1-p)\left(\frac{a+b}{2}
\right)$ come up heads.

\hrulefill

\begin{enumerate}
\item Transmit $x_i$

Let $Y$ be the number of 1's received.

\item $\beta_i:=\begin{cases}
1 &\text{if number of 1's received is greater than $(1-p)\left(\frac{a+b}{2}\right)$}\\
0 &\text{otherwise}
\end{cases}$
\item \label{step:comparison} Transmit $\beta_i$
\item Let $\widehat{\theta}_{s,i}$ be the fraction of received
bits from Step \ref{step:comparison} that are 1 (i.e., the total number of 1's received, divided by the total number of 1's or 0's received).
\item \label{step:estimatewt} $\widehat{s}_i=\arg\min_{\ell}|\theta_\ell-\widehat{\theta}_{s,i}|$
\item \label{step:encoding2}
Let $i'=i~\mathrm{mod}\log n$ and transmit $\mathsf{Enc}(\widehat{s}_i)_{[K(i'-1) + 1, Ki']}$
over $K$ rounds
\item \label{step:decoding} $\tilde{s}_i=
\begin{cases}
\text{decoded string} &\text{if at least $.88K\log n$ bits were  received in Step \ref{step:encoding2}}\\
\widehat{s}_i &\text{if fewer than $.88K\log n$ bits were  received in Step \ref{step:encoding2}}
\end{cases}$
\end{enumerate}
}
\end{algorithm}

\begin{lemma}
On running Algorithm \ref{algo:pinpointweight}, all processors return the
Hamming weight $s$ of $X$ with probability at least $0.9$.
\end{lemma}
\begin{proof}
Define $\widehat{\theta}_s$ to be the fraction of $\beta_i$ transmitted in Step \ref{step:comparison}
that are 1.

We can lower bound $\theta_{\ell+1} - \theta_\ell$ for $x\leq\ell < y$ by $\frac{c}{\sqrt{n}}$
where $c$ is some constant \cite[Lemma 41]{goyal2008lower}. The probability that $|\theta_s-
\widehat{\theta}_s|$ is at most $\frac{c}{8\sqrt{n}}$ can be made at least $0.99$ with an appropriate
choice of the number of repetitions $\gamma$. Similarly, we can ensure that $|\widehat{\theta}_s-
\widehat{\theta}_{s,i}|$ is at most $\frac{c}{8\sqrt{n}}$ with probability at least $0.99$.

By Chernoff bounds, the fraction of processors for which $|\widehat{\theta}_s-\widehat{\theta}_{s,i}|<\frac{c}{8\sqrt{n}}$
is at least $0.95$ with probability at least $1-\exp(-\Omega(n))$. Thus, conditioned on
$|\theta_s-\widehat{\theta}_s|<\frac{c}{8\sqrt{n}}$, we have that for at least $0.95$ of the processors,
$|\theta_s-\widehat{\theta}_{s,i}|<\frac{c}{4\sqrt{n}}$. Further, the string $S_1S_2\ldots S_n$ transmitted
in Step \ref{step:encoding2} with random erasures has distance less than the decoding radius of
$\mathcal{C}$ of $\mathsf{Enc}(s)$ with probability at least $1-\exp(-\Omega(n))$, in which case
all processors can correctly output $s$.

Since the condition $|\theta_s-\widehat{\theta}_s|<\frac{c}{8\sqrt{n}}$ holds with probability at
least $0.99$, the required guarantees of the Lemma hold.
\end{proof}

\begin{proof}[Proof of Theorem \ref{thm:symmetricprotocol}]
The processors run Algorithm \ref{algo:determineinterval} to obtain 3 candidate intervals $I_1,I_2$
and $I_3$, and with exponentially high probability, at least two of these candidate intervals contain the Hamming
weight. The processors run Algorithm \ref{algo:pinpointweight} on each of the three intervals
and processor $i$ obtains outputs $n_0,n_1$ and $n_2$ respectively. With constant probability, at
least two of $n_0,n_1$ and $n_2$ are the same and equal to the correct Hamming weight, and hence outputting
the majority of the three matches the guarantee.
\end{proof}

\end{document}